\documentclass[a4paper,12pt,article]{iopart}
\usepackage{currfile}
\usepackage{graphicx,caption}
\usepackage{amssymb}
\usepackage{amstext}
\usepackage{bm}
\usepackage{epstopdf}
\usepackage{color}
\usepackage{mhequ}
\usepackage{times}
\usepackage{amsthm}
\usepackage{accents}
\usepackage{enumitem}
\usepackage[titletoc,title]{appendix}
\usepackage{perpage} 
\MakePerPage{footnote} 
\usepackage{booktabs}
\usepackage{sub_JP}

\def\Label#1{}
\def\const{{\rm const.}}

\let\phi=\varphi
\let\kappa=\varkappa

\usepackage{chngcntr}
\counterwithin*{equation}{section}

\renewcommand{\theequation}{\thesection.\arabic{equation}}
\renewcommand{\thesection}{\arabic{section}}
\renewcommand{\thesubsection}{\thesection.\arabic{subsection}}

\def\d{{\rm d}}

\let\epsilon=\varepsilon
\let\theta=\vartheta

\let\rho=\varrho

\def\DDR{{\cal D^{\real}\kern-0.8em}}

\def\real{{\bf R}}
\def\lref#1{Lemma~\ref{#1}}
\def\aref#1{Appendix~\ref{#1}}

\def\fref#1{Fig.~\ref{#1}}
\def\cref#1{Corollary~\ref{#1}}
\def\pref#1{Proposition~\ref{#1}}
\def\tref#1{Theorem~\ref{#1}}

\def\sref#1{Sect.~\ref{#1}}

\newtheorem{theorem}{Theorem}[section]
\newtheorem{proposition}[theorem]{Proposition}
\newtheorem{lemma}[theorem]{Lemma}

\newtheorem{remark}[theorem]{Remark}
\newtheorem{conjecture}[theorem]{Conjecture}

\definecolor{mydarkgreen}{rgb}{0.0, 0.5, 0.0}

\newcommand{\fh}[1]{\textcolor{black}{#1}}
\def\HALF{{\textstyle\frac{1}{2}}}
\def\TWOTHIRDS{{\textstyle\frac{2}{3}}}
\def\FOURTHIRDS{{\textstyle\frac{4}{3}}}
\def\citep#1{\cite{#1}}

\begin{document}
\null
\title[Instabilities in effective field theories]{Instabilities Appearing in \fh{Cosmological} Effective Field theories: When and How?}

\author{Jean-Pierre Eckmann$^1$, Farbod Hassani$^{2,3}$, Hatem Zaag$^4$}

\address{$^1$ D\'epartement de Physique Th\'eorique and Section de
  Math\'ematiques, University of Geneva, Switzerland}
  \address{$^2$ Institute of Theoretical Astrophysics, Universitetet i Oslo, 0315 Oslo, Norway}
   \address{$^3$ Department of Physics, McGill University, 3600 rue University, Montreal, QC H3A 2T8, Canada}

\address{$^4$ Universit\'e Sorbonne Paris Nord ,
LAGA, CNRS (UMR 7539), F-93430, Villetaneuse, France}

\begin{abstract}
Nonlinear partial differential equations appear in many
  domains of physics, and we study here a typical equation which one
  finds in effective field theories (EFT) originated from cosmological studies. In particular, we are
  interested in the equation $\partial_t^2 u(x,t)  =  \alpha (\partial_x u(x,t))^2 +
  \beta \partial_x^2 u(x,t)$ in $1+1$ dimensions. It has been
  known for quite some time that solutions to this equation
  diverge in finite time, when $\alpha >0$. We study the
  nature of this divergence as a function of the parameters $\alpha>0 $ and $\beta\ge0
  $. The divergence does not disappear even when $\beta $ is very
  large contrary to what one might believe \fh{(note that since we consider fixed initial data, $\alpha$ and $\beta$ cannot be scaled away).} But it will take longer to appear as $\beta $ increases when $\alpha$ is fixed. We note that there are two types of divergence and we discuss the transition between these two as a function of parameter choices.
  The blowup is unavoidable
  unless the corresponding equations are modified.
  Our results extend
  to $3+1$ dimensions.
\end{abstract}
\submitto{\NL}

 \section{Introduction}
 In physics, effective field theories (EFT) are employed to
 describe fundamental theories at the low energy limit in a unified form.  This approach is widely used to express different
 physical phenomena \cite{Cheung:2007st,Pich:1998xt,Gubitosi_2013,Goldberger:2004jt,Endlich:2012vt}. This framework, which is constructed based on perturbative
 expansion, usually leads to non-linear partial 
 differential equations. 
 
Recently, the effective field theory approach has become
very popular in cosmological studies, especially to study the late time
accelerating expansion of the Universe which is driven by the so-called
dark energy component \citep{Amendola:2016saw,2012PhR...513....1C}. Based
on cosmological observations \citep{Ade:2015xua,2018ApJ...859..101S,2017MNRAS.470.2617A} the clustering of the dark
energy component is supposed to be small, so the linear approximations are justified \citep{Zumalacarregui:2016pph,Hu:2013twa}. However, in the near future high
precision measurements of the Universe will be done by the new
cosmological surveys \citep{Amendola:2016saw,Santos:2015gra,4MOST:2019,Aghamousa:2016zmz}. This motivated cosmologists \citep{Hassani:2020agf,Hassani:2020buk,Cusin:2017mzw} to study non-linear
PDEs arising from these EFT approaches to have more accurate predictions of these theories. Thus, cosmological $N$-body
simulations have been developed \citep{Hassani:2019lmy, Adamek:2015eda} which describe the evolution of structures in the Universe by solving Einstein field equations
 as well as a non-linear PDE for the dark energy 
component. In a nutshell, this non-linear PDE has non-linearities that
sometimes are dominated by a $(\partial_x u(x,t))^2$ term \citep{Hassani:2021tdd}. This motivates our current study. 
 
Using extensive numerical simulations with
$k$-evolution \citep{Hassani:2019lmy} it was discovered earlier in \citep{Hassani:2021tdd} that the solutions of
such equations can form violent singularities at \emph{finite}
time. And, depending on the cosmological parameters, these
singularities can even appear
at a time \emph{before} the current epoch of the Universe. Obviously,
this asks for a change of parameters, or for regularising the models
with additional smoothing terms. In \citep{Hassani:2021tdd} it is specifically argued that appearance of the Laplace term $\beta \partial_{xx} u(x,t) $ with large enough $\beta$ makes the system stable. \fh{It is worth noting that in cosmological context $\sqrt{\beta}=c_s$ where $c_s$ is the ``speed of sound" and determines how fast the perturbations of scalar field propagate. In cosmology usually, $0\le c_s \le 1$. However superluminal cases ($c_s \ge 1$) are also considered in the literature \cite{Babichev:2007dw, Bonvin:2006vc}.} However, we will show that the
Laplacian term does not regularize enough to avoid the finite time blowup:
It just shifts the blowup time to a later epoch. 
Perhaps, not seeing the instability in the realistic cosmological $N$-body simulations might be due to the short time period, or some other phenomena which are present in cosmological setups.

In the paper \citep{Hassani:2021tdd}, \fh{employing the cosmological $N$-body code $k$-evolution,} it was found that the PDE for the \fh{EFT of}
dark energy (in particular, $k$-essence models\footnote{ \fh{These are a general class of theories in which the action contains at most one temporal and one spatial derivative acting on the field \cite{Armendariz-Picon:2000ulo}. The $k$-essence models have been proposed as a possible explanation for the late time cosmic acceleration. In the $k$-evolution code, the $k$-essence field and other cosmological components, such as dark matter and baryonic matter, are implemented. The ``$k$'' in ``$k$-essence'' stands for ``kinetic," which refers to the kinetic energy of the theory. The idea behind $k$-essence is that the negative pressure of its fluid description is caused by its kinetic energy.}}), for some set of parameters
leads to an instability in finite time. It was shown in
\citep{PanShi_Hamilton} that the main
source of finite-time instability is due to the presence of the non-linear term in
  \begin{equ}\label{eq:pidd}
\frac{\partial^2 u(x,t)}{\partial t^2}  =\alpha \cdot (\frac{\partial u(x,t)}{\partial x})^2~,
  \end{equ}
which  appears naturally in EFT theories (in cosmological studies $u(x,t)$ is called $\pi(x,t)$). 
Here, we consider only fixed  $\alpha>0 $
but for general theories $\alpha $ can be time dependent. This time dependence appears in $k$-essence theories through the Hubble parameter $\mathcal H(t)$ \citep{Hassani:2021tdd} and we
have checked that the divergence persists when taking into account the time dependence of $\mathcal H$.

\fh{Using a contraction-mapping fixed point argument, one can easily check that the (local in time) Cauchy problem for equation \eref{eq:pidd} can be solved in the space consisting of Fourier transforms of compactly supported functions. Beyond this setting, }
it is easy to see that \eref{eq:pidd} has
solutions which diverge in finite time \citep{Hassani:2021tdd}:
In fact, for some initial conditions we can consider
$u$   of the form $u(x,t)=  f(t)x^2$. This
leads to $f''(t)=4 \alpha  \cdot f (t)^2$. 
When the initial condition is $f(0)=\frac{3}{2\alpha t_0^2}$ and
$f'(0)=\frac{9}{\alpha t_0^4}$, then we have
\begin{equ}\label{eq:aaprime}
  f(t)=\frac{3}{2 \alpha  (t_0-t)^2}~,
\end{equ}
which diverges as $t\uparrow t_0$,
when $\alpha >0$. As discussed in \citep{PanShi_stability,PanShi_scale_invariant,PanShi_Hamilton},
this divergence is of a local type and the curvature of the minima
increases to infinity in a finite time. Here we call this divergence the ``V''
type as shown in \fref{fig:u}. We give a general divergence proof  in \sref{sec:blowup}.
  \begin{figure}[ht!]
\centering\includegraphics[width=1.\textwidth]{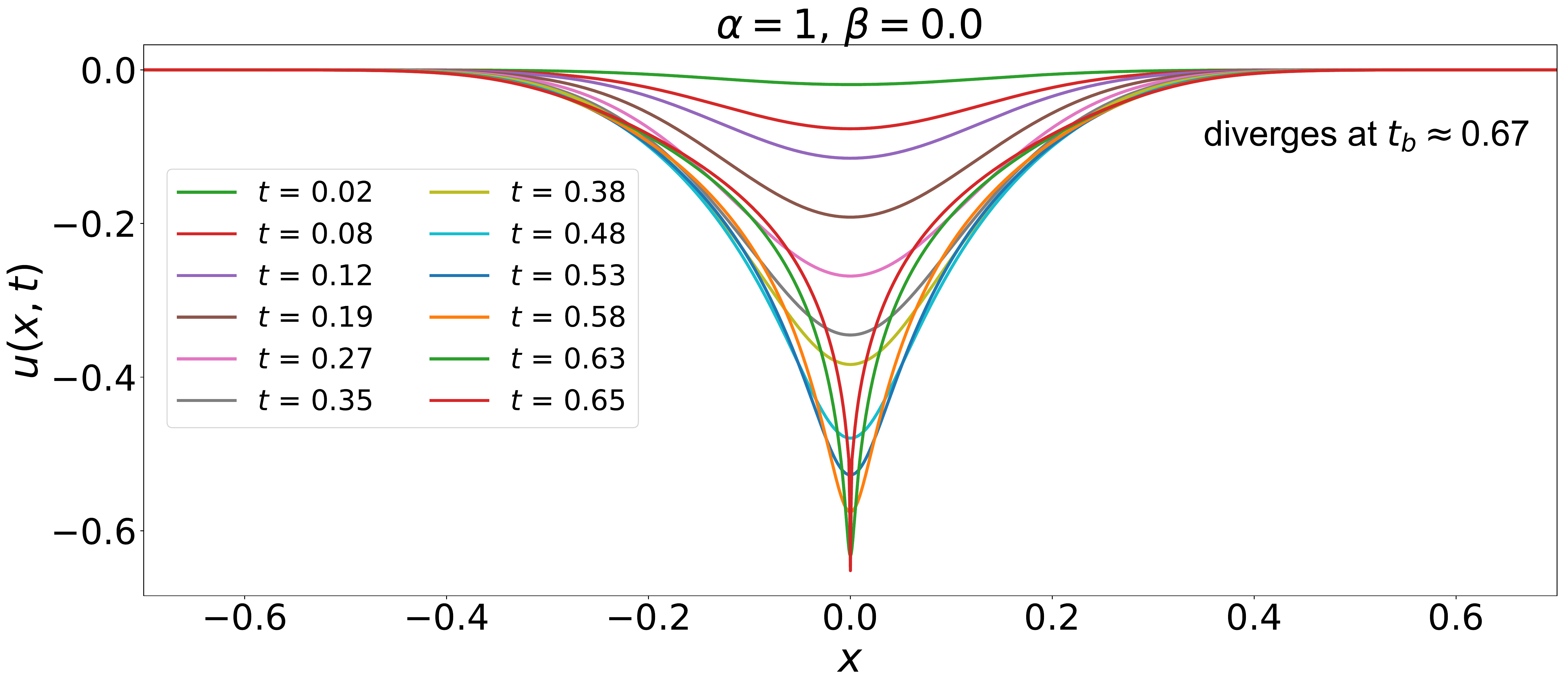}
  \caption{{\bf{The V-type divergence}}: The solution profile $u(x,t)$ is shown at different times (different colours), up to the blowup time. The curvature at the minimum diverges in
    finite time. The initial conditions are: $u(x,0)=0$ and
    $\partial_t u(x,0)=-\exp(-120 \pi^2 x^2)$ on
    $[-\pi,\pi]$.}
  \label{fig:u}
  \end{figure}

  Consider now the equation
    \begin{equ}\label{eq:ab}
\frac{\partial^2 u(x,t)}{\partial t^2}  =\alpha \cdot (\frac{\partial u(x,t)}{\partial
  x})^2+\beta \frac{\partial^2 u(x,t)}{\partial
  x^2}~.
  \end{equ}
When $\beta =0$ this is \eref{eq:pidd}. 
We  next consider the case when $\beta >0$ (and $\alpha >0$ is
  fixed).
\fh{The local in time Cauchy problem for this equation can be solved in $W^{1,\infty}\times L^\infty(R)$ (see \cite{azaiez_masmoudi_zaag_2019} for a similar strategy). 
From 
the finite speed of propagation, for given initial data at $t=0$, two cases appear for the domain of definition of the solution:\\
- either it is $\{(x,t)\;|\;t\ge 0\}$ and the solution is said to be global;\\\
- or it can be expressed as
\[
\mathcal{D}=\{(x,t);|\;0\le t <T(x) \}
\]
for some Lipschitz function $x\mapsto T(x)$ (with Lipschitz constant $1/\sqrt\beta$), where the solution is said to "blow up" in finite time. (see Appendix \ref{sec:appb} for details). In this paper, we deal with blow-up solutions.}

\medskip

\fh{As for the blow-up behaviour,}
there are two scenarios: When $\beta$ is very small,
  the solution will be of V-type but when  $\beta $ is larger, then it
  will be of a shape we call M-type, as illustrated in \fref{fig:front}.

\begin{figure}[h!]
  \centering
  \includegraphics[width=1.\textwidth]{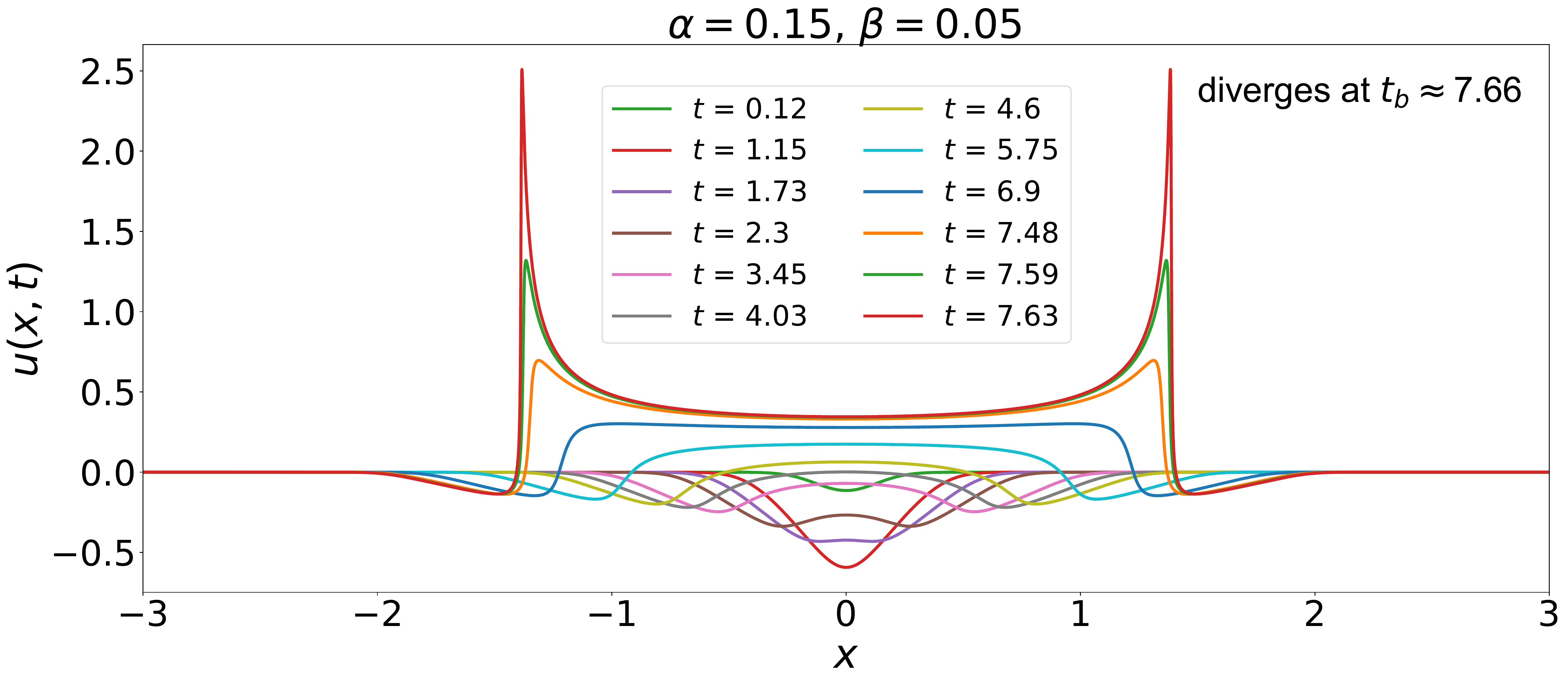}
  \caption{{\bf{The M-type divergence}}: As time advances,
    the support of the function spreads, and the function at the edge
    gets steeper, until the derivative will diverge at some time
    $T_*$.
    The simulation is for the equation \eref{eq:ab}, with
    $\alpha=0.15$ and $\beta =0.05$ and initial conditions
    $\partial_t u(x,0)=-\exp(-120 \pi^2 x^2)$, $u(x,0)=0$ on
    $[-\pi,\pi]$.
  }\label{fig:front}
\end{figure}

It was seen in simulations presented in \citep{Hassani:2021tdd}, that the
instability seems to vanish when one adds a Laplace term with large enough coefficient $\beta$.

However, this is not the whole story:
In fact, whenever $\alpha> 0$ this system is
always unstable and even large $\beta$ cannot cure the
instability forever, (this holds rigorously for a large class of initial conditions, and seems to hold in numerics, for \emph{any} initial condition with compact support). However, increasing $\beta$ does increase the
blowup time. Actually, the blowup now happens at the ``ends'' of the
``M" shape in the profile, whose walls get steeper and steeper, until the derivative becomes
infinity. This phenomenon was known for some time in the literature
\citep{rammaha,john}, and it is quite generic.
In simulations, perhaps one does not wait long enough, or the M gets
too wide in the numerically allocated spatial direction, before the divergence happens.

In the following sections, we will study in more detail the domains in
the $\alpha ,\beta $ plane 
for which ``V" and ``M" divergences will happen and the time it takes for
blowup to occur.
This will tell us for which
physical parameters the singularity is so far in the future that it
can be neglected, or that it is of a form which can easily be damped by
adding 
additional terms to \eref{eq:ab}. \fh{However, such terms do not seem to occur naturally in the case of cosmological models for dark energy, especially within the weak-field approximation \cite{Hassani:2020agf}.}
For the convenience of the reader, we repeat in two appendices some
details about the V type divergence, and we also repeat---with small
variations---a proof of the persistence of divergence for all $\beta >0$
(when $\alpha >0$). These appendices are based on \citep{PanShi_scale_invariant}
and \citep{rammaha,rammaha2}.

\section{Blowup time as a function of $\alpha $ when $\beta =0$}\label{sec:blowup}

Before we can study the dependence on $\beta $, we need to 
study the divergence time for the case $\alpha >0$.
The following is a slight adaptation of the results of \citep{PanShi_scale_invariant,PanShi_Hamilton}.

We consider the equation $u_{tt}=\alpha (u_x)^2$ on the real line.
We start by writing the solution in the form
\begin{equ}\label{eq:peteralpha}
  u(x,t)=f(x)+g(x)t +\alpha \int_0^t \d\tau \int_0^\tau  \d\tau'
  (u_x(x,\tau '))^2~. 
\end{equ}
This corresponds to the initial conditions
\begin{equ}
  u(x,0)=f(x)~,\quad  u_t(x,0)=g(x)~.
\end{equ}
We will consider the case where $f'(0)=g'(0)=0$, and we ask how the
solution behaves near $x=0$. Depending on the curvatures $f''(0)$ and
$g''(0)$, the second derivative $u_{xx}(x,t)$ will, or will not
  diverge at $x=0$. Of course, if the functions $f$ and $g$ have
  vanishing derivatives at some other point(s) $x_0$, the same
  discussion will apply at those points, and there can be one of these
  points where $u_{xx}(x_0,t)$ diverges before the one at $x=0$. In
  the following proposition, we will neglect this aspect.

\begin{proposition}\label{prop:peteralpha}
  Assume $f'(0)=g'(0)=0$. Define
  \begin{equ}\label{eq:smallc}
    c=\HALF g''(0)^2-\TWOTHIRDS \alpha f''(0)^3~. 
  \end{equ}

  Then the following cases appear:\\
  (i) If $g''(0)>0$ then $u_{xx}(0,t)$ diverges in finite time $t_+$
  given by
  \begin{equ}\label{eq:divalpha}
    t_+=\int_{f''(0)}^\infty \frac{\d b }{\sqrt{\FOURTHIRDS\alpha  b^3
        +2c}}=\int_{f''(0)}^\infty \frac{\d b }{\sqrt{\alpha (\FOURTHIRDS  (b^3 -f''(0)^3)+g''(0)^2}}~.
  \end{equ}\\
  (ii) If $g''(0)<0$ then $u_{xx}(0,t)$ will converge to $b_*$ in a
  finite time $t_-$, where
  \begin{equ}\label{eq:convalpha}
    \TWOTHIRDS \alpha b_*^3 =-c~,\text{ and }
    t_-=\int_{b_*}^{f''(0)} \frac{\d b}{{\sqrt{\FOURTHIRDS \alpha  b^3 +2c}}}~.
  \end{equ}
  At this point in time, we will have  $u_t(0,t_-)=0$
  which corresponds to (iii) and the solution will diverge after
  another finite time $t_+$ (unless $c=0$).
  \\
  (iii)  If  $g''(0)=0$, and $f''(0)\ne0$ then $u_{xx}(0,t)$ diverges
  in finite time $t_+$ given again by \eref{eq:divalpha}.\\
  (iv) If $g''(0)=0$ and $f''(0)=0$ then $u_{xx}(0,t)$ stays constant.
\end{proposition}
\begin{remark} \label{rem2d2}
  When $g''(0)>0$ and $f''(0)=0$,  then the elliptic integral can be evaluated
  explicitly and one gets
  \begin{equ}
    t_+(\alpha ,g''(0))=\frac{A}{\alpha ^{1/3} g''(0)^{1/3}}~,
  \end{equ}
  with $A \approx 2.5479$.
\end{remark}
 \begin{remark}
  Assume that $u(x,0)=0$ and $u_t(x,0)$ is  a smooth, bounded function  with
  \fh{several} well-separated extrema. Such initial conditions are typical for
  questions in cosmology.
In this case, the blowup will happen first in \fh{that} point $x_0$ for
which $C\equiv \alpha \, u_{txx}(x_0,0)$ is maximal (and $u_{tx}(x_0,0)=0$).
 In that case, \fh{the formula of Remark.~\ref{rem2d2} leads to} $t_+\sim 2.547/C^{1/3}$.
\end{remark}
\begin{remark}
  Note that $t_++t_-$ is the total time for an initial condition
  $g''(0)<0 $ to diverge (only when $c\ne0$, with $c$ defined in \eref{eq:smallc}). And then, the divergence
  time is $t_-(f''(0), g''(0))+t_+(b^*)$.

\end{remark}

The proof of these statements is given in \aref{a:hatem}.

\section{Crossover}\label{sec:crossover}

\begin{figure}[ht!]
\centering\includegraphics[width=1.\textwidth]{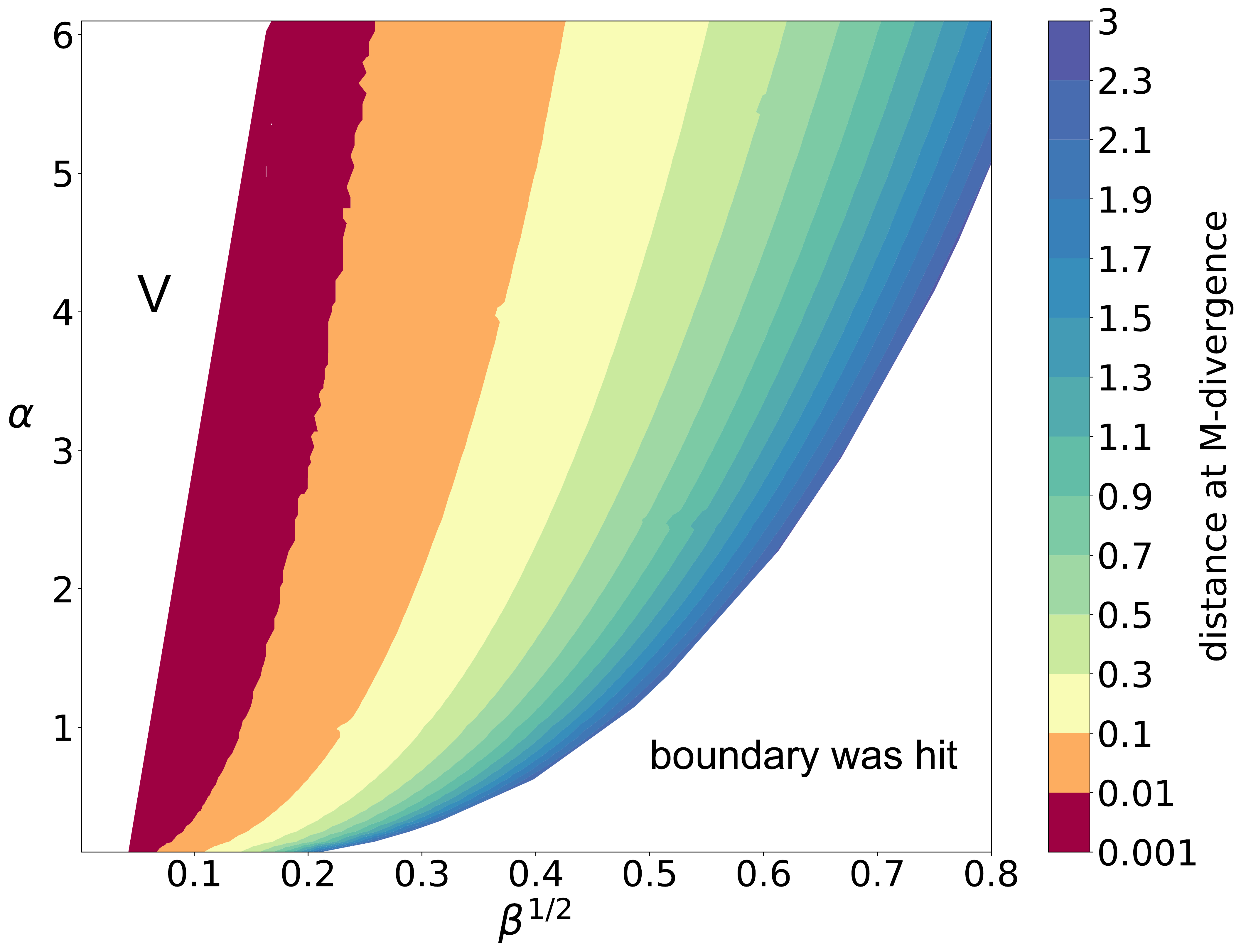}
  \caption{{\bf{Distance of singularity at divergence time as a function of  $\alpha $ and $\beta^{1/2} $:
    }}
    In the region
marked ``V'' the numerics does not allow to distinguish between V-type
and M-type divergence, since the distance of the (two) singularities
is too small to distinguish M and V.
The other white region is the set of parameters where the simulation hits the boundary before divergence (see \aref{a:numerics} for an explanation). The initial condition is $u(x,0)=0$ and $u_t(x,0)=-\exp(-120 \pi^2 x^2)$. The colour coded part shows the distance of the maximum of $|u_{xx}|$  at blowup time.
  }\label{fig:smallbetadist}
  \end{figure}

\begin{figure}[ht!]
\centering\includegraphics[width=1.\textwidth]{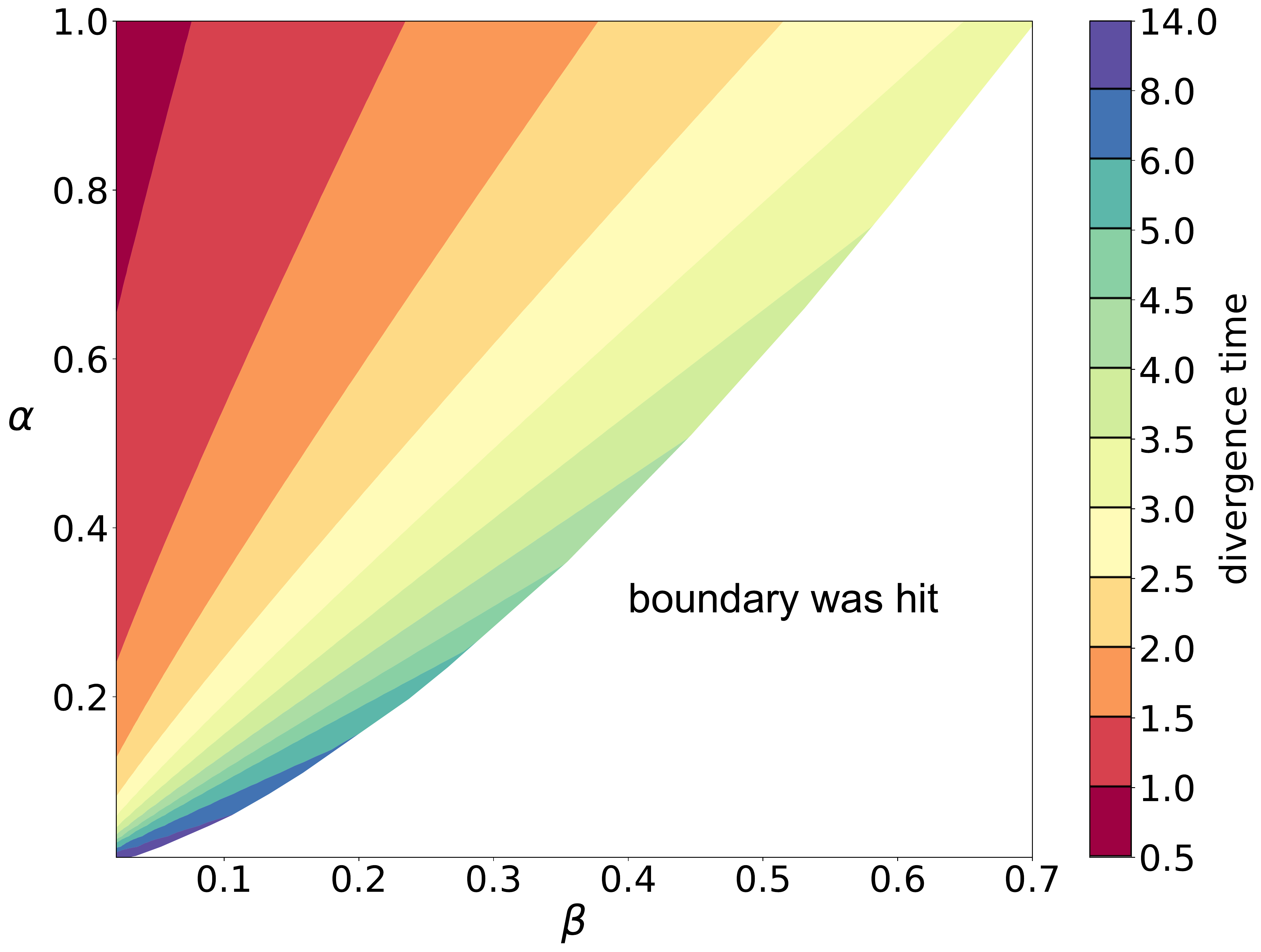}
  \caption{{\bf{Divergence time as a function of $\alpha $ and $\beta $ :}} The $x$ axis is $\beta$. The initial condition is $u(x,0)=0$ and $u_t(x,0)=+\exp(-120 \pi^2 x^2)$. The white region contains
       those  parameters ($\alpha $ and $\beta $) for which the solution hits the
        boundary before divergence, as in \fref{fig:smallbetadist}. Each colour corresponds to a different divergence
        time. Note that in this case, since the initial condition satisfies the condition of \tref{thm:t1} (we made the support finite), we \emph{know} that the solution must blow up in finite time, for all $\alpha>0$ and $\beta\ge0$.}
  \label{fig:smallbeta}
  \end{figure}

\begin{figure}[ht!]
\centering\includegraphics[width=1.\textwidth]{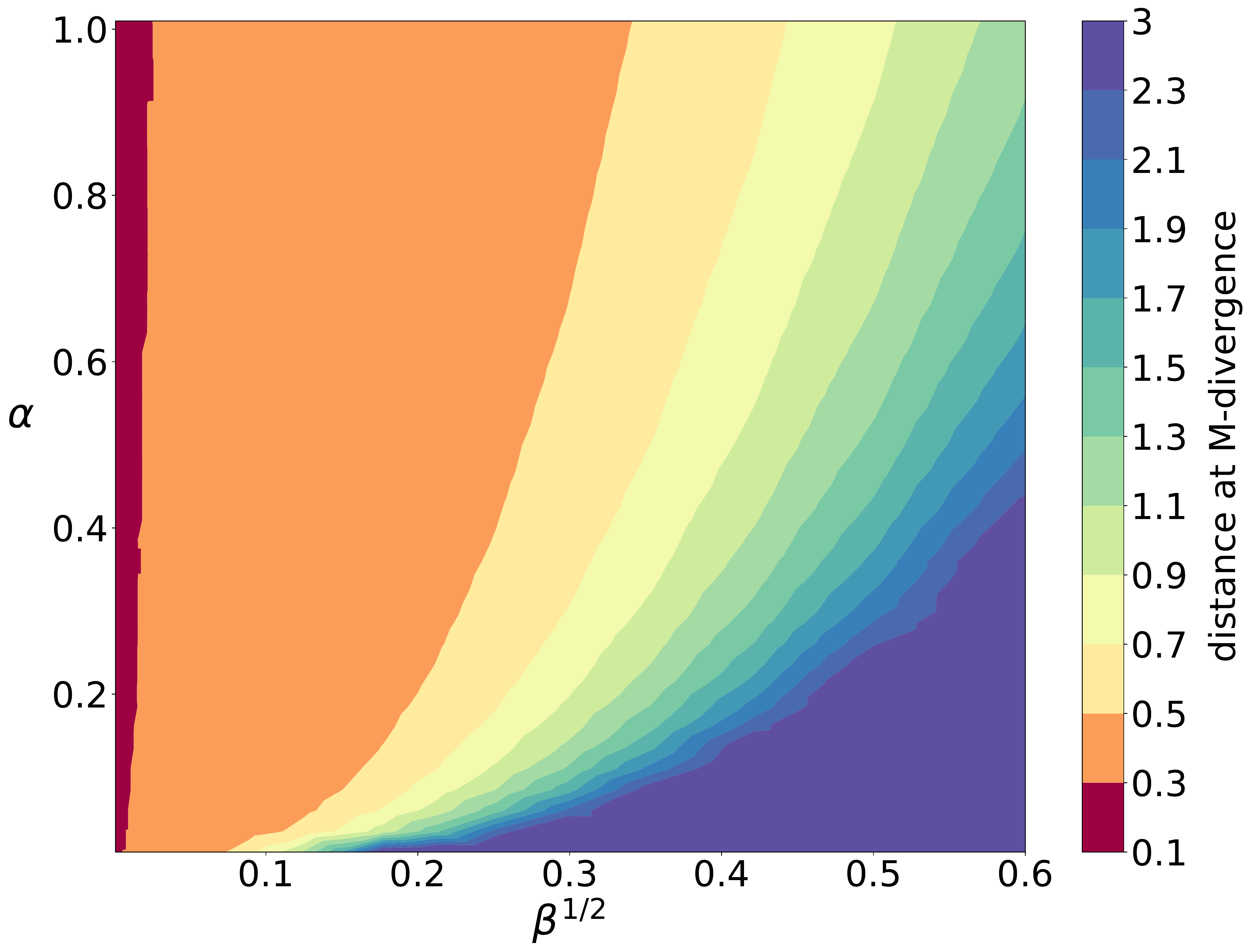}
  \caption{{\bf{Distance of singularity at divergence time as function of  $\alpha $ and $\beta^{1/2} $: 
        }} Because the initial condition is now $u(x,0)=0$ and $u_t(x,0)=+\exp(-120 \pi^2 x^2)$,
    \emph{all} divergences seem to be of M-type (see \tref{thm:t1}, which only asserts divergence in finite type, but the proof suggests divergence at near the advancing front). Note that for all
    $\alpha $ and $\beta>10^{-6}$ considered in this graph, the distance 
    of
    the singularity from  $x=0$ at blowup is at least $0.1795$ (in the interval    $[-\pi,\pi]$). For this initial condition and the $\alpha$, $\beta$, considered, blowup always happens before the wave hits the boundary.}
    \label{fig:allm}
  \end{figure}

\fh{The equation \eref{eq:ab} can be rescaled in space and time, to eliminate either $\alpha$ or $\beta$. This rescaling will be used in \aref{sec:appb}, when we show that blowup is unavoidable in many cases. Here, we ask a different question as in cosmological studies it is important to know the blowup as a function of $\alpha$ and $\beta$ for \emph{a fixed} initial perturbation.  And then, the scaling cannot be used.}

As we have indicated in the introduction, for a fixed $\alpha >0$ there
will be a crossover between the V-type and the M-type divergence as
one varies $\beta $. This crossover is important in general, because the nature of the divergence is different in the two cases:
For the V-type divergence, we see a \emph{localised} blow-up of the
second derivative in a minimum, as shown in \fref{fig:u}. In the second,
M-type case, the divergence happens at some distance from the critical
point of the initial condition (see \fref{fig:front}). A vertical wall forms, and this wall
moves outward from the centre, basically with the propagation speed
of the wave (which is $\sqrt{\beta }$ for \eref{eq:ab}). 

The divergences in the EFT framework are mainly important, because
they are either a hint for entering the strong field regime (where the
perturbative expansion is not valid anymore) or a breakdown of the
underlying fundamental theory \citep{Hassani:2021tdd}. The divergence type
should correspond to physical phenomena happening at high energy/small
scales. 
Especially it may help to introduce appropriate mechanisms to remove
such instabilities.  As an example the V-type divergence is localized,
and in cosmological studies (as suggested in \citep{Hassani:2021tdd}) can be
used as an origin of super massive black holes. On the other hand the
M-type blowup resembles caustics formation in the Universe
\citep{Babichev:2017lrx,Babichev:2016hys}. These two cases have different signatures
in cosmological observables \citep{Sawicki:2013wja}. And therefore it is useful to
distinguish them. 

We illustrate the various possibilities in Figs.~\ref{fig:smallbetadist}--\ref{fig:allm}.
In \fref{fig:smallbetadist} we consider
initial conditions with a minimum, namely $-\exp(-120 \pi^2 x^2)$. In this
case, depending on $\alpha $ and $\beta $ we can see a V-type
divergence or an M-type divergence. When $\alpha $ is large and $\beta
>0$ not too small, we detect clearly an M-type divergence. But when
$\beta $ is too close to 0, the numerics breaks down, and  one
sees something like a V-type divergence in the region ``V'' of
\fref{fig:smallbetadist}. However, it is also possible that in fact
for all $\beta >0$ one has M-type divergence, with the two walls of
the M so close together that the numerics gets unreliable. With our current
understanding we conjecture that it is always M-type when $\beta
>0$. Note that, in any case \tref{thm:t1} shows that for all $\beta
>0$ there is finite time divergence (the proof only applies to initial conditions of a certain positivity type---the function $M$ of \aref{sec:appb}, but we have seen the M-type divergence for general non-trivial initial conditions with compact support).

Indeed, when $\alpha $ is large
and $\beta >0$ is small, then the minimum will get more and more
pointy as in \fref{fig:u}, until the second derivative diverges. This happens
in the white region ``V'' in \fref{fig:smallbetadist}. In the coloured parameter
regions, the negative initial conditions get more pointy, until the
Laplacian term regularises the central part, which then grows,
becomes positive and then diverges away from the centre, as in \fref{fig:front}.

When the initial condition is positive, with a local maximum, then the
situation is somewhat simpler, because the transition from negative to
positive is absent. No V singularity can form. Note that this is
consistent with the discussion of the cases (i) and (ii) in
\pref{prop:peteralpha}. The divergence time as function of $\alpha $
and $\beta$ is shown in \fref{fig:smallbeta} for the divergence time,
and in \fref{fig:allm} for the divergence distance, by which
we mean the distance of the singularity from the coordinate
origin. We can not offer a formula for the curves in
either of the figures.

\section{Conclusions and discussion}\label{sec:Conclusions}
In this article we discussed the equation $\partial_t^2 u(x,t)  =  \alpha (\partial_x u(x,t))^2 +
  \beta \partial_x^2 u(x,t)$. These types of PDEs appear naturally in
  the effective field theory descriptions of physical systems, where
  one approximates the equations assuming weak fields. We specifically
  discuss the two divergence types arise from these equations, namely
  the ``V" and the ``M" type. We discuss how and when these
  instabilities are generated. In some cosmological studies
  \cite{hassani_JCAP} it is suggested that the term $\beta \partial_x^2
  u(x,t)$ stabilises the system so that the instability caused by $
  \alpha (\partial_x u(x,t))^2$ vanishes. This stabilisation is
  especially motivated by realistic cosmological studies. While our
  results are in agreement with \citep{Hassani:2021tdd,PanShi_Hamilton,hassani_JCAP} regarding
  the V-type blowup which happens for small $\beta$, our results show
  that even for large $\beta$ the instability is unavoidable, however
  it is always of the new M-type.
  The reasons that this instability is
  not seen in the realistic cosmological simulations for large $\beta$
  (or $c_s^2$ in cosmology)  could probably be due to 1) the role of
  gravity which is neglected in our study 2) the difference in
  boundary conditions or 3) the instability exists but will appear
  beyond the times considered in a study. Preliminary studies show
  that the results carry over to the 3+1 dimensional
  situation.\footnote{In the $3+1$ dimensional case with spherical
    symmetry, the equation \eref{eq:ab} is simply replaced by
     $u_{tt}(r,t)=\alpha (u_r(r,t))^2 +\beta
    \left(u_{rr}(r,t)+\frac{2}{r}u_r\right)$, where $r$ is the radial
    coordinate. We have done the corresponding numerical experiments
    for this case.}
\newpage

\section*{Appendices}
\begin{appendices}

\section{Divergence time when $\beta =0$}\label{a:hatem}

Here we prove \pref{prop:peteralpha}.
\begin{proof}

Define 
\begin{equa}
  a(t)= u_x(0,t)~,\quad b(t)=u_{xx}(0,t)~.
\end{equa}
Then \eref{eq:peteralpha} leads to
\begin{equa}
  a(t)&=f'(0)+g'(0)t+2\alpha \int_0^t \d\tau \int_0^\tau  \d\tau' u_x(0,\tau ')
  u_{xx}(0,\tau ')~,
\end{equa}
and therefore
\begin{equa}\label{eq:aalpha}
  \ddot a(t)=2\alpha  a(t)\,b(t)~.
\end{equa}
Similarly,
\begin{equa}
  b(t)=f''(0)+g''(0)t +2\alpha \int_0^t \d\tau \int_0^\tau  \d\tau'
\left(  (u_{xx}(0,\tau '))^2
  +u_{x}(0,\tau ')u_{xxx}(0,\tau ')\right)~.
\end{equa}
From this, we deduce
\begin{equa}[eq:balpha]
  \ddot b(t)&=
  2\alpha (u_{xx}(0,t))^2
  +2\alpha u_{x}(0,t)u_{xxx}(0,t)\\
  &= 2\alpha b(t)^2 + 2\alpha a(t)u_{xxx}(0,t)~.
\end{equa}
Since we assume $f'(0)=g'(0)=0$ we find from \eref{eq:aalpha} that
$a(t)=0$ for all $t$ for which $b(t)$ is finite. Therefore,
\eref{eq:balpha} reduces to
\begin{equa}\label{eq:bddalpha}
\ddot b(t)=2\alpha (b(t))^2~.  
\end{equa}
We will discuss this equation. For computing the divergence time, it
is useful to transform the equation as follows: Multiplying by $\dot b$
leads to
\begin{equa}
  \HALF \frac{\d}{\d t}(\dot b(t))^2 = \TWOTHIRDS \alpha \frac{\d }{\d
  t}b(t)^3~,
\end{equa}
or, for some $c$,
\begin{equa}\label{energyalpha}
  \HALF(\dot b(t))^2  = \TWOTHIRDS\alpha (b(t))^3 +c~.
\end{equa}
Note that looking at $t=0$ we find
\begin{equ}\label{eq:calpha}
  c=\HALF \dot b(0)^2 -\TWOTHIRDS\alpha  b(0)^3= \HALF  g''(0)^2
  -\TWOTHIRDS\alpha  f''(0)^3~,
\end{equ}
which is the definition in the proposition. Note that
\begin{equa}
  b(0)=f''(0) ~, \quad \text{and}\quad \dot b(0)=g''(0)~.
\end{equa}

We consider first the case where
$g''(0)>0$. Then $\dot b(0)>0$ and from \eref{energyalpha} we find that
\begin{equa}\label{eq:dotbalpha}
  \dot b(t)  =\sqrt{\FOURTHIRDS\alpha  (b(t))^3 +2c}~,
  \end{equa}
which means $b$ is increasing and the quantity below the square
  root 
  is always positive. Using standard techniques, we get 
\begin{equa}
  \d t  = \frac{\d b}{\sqrt{\FOURTHIRDS\alpha  b^3 +2c}}~.
\end{equa}
From \eref{eq:dotbalpha} we deduce 
the divergence time $t_+$,
\begin{equa}\label{eq:t2alpha}
  t_+=\int_{b(0)}^\infty \frac{\d b }{\sqrt{\FOURTHIRDS\alpha  b^3 +2c}}~.
\end{equa}
This proves \eref{eq:divalpha}.

The case $g''(0)<0$ is handled similarly, but now \eref{eq:dotbalpha} is
replaced by 
\begin{equa}\label{eq:dotbminusalpha}
  \dot b(t)  =-\sqrt{\FOURTHIRDS\alpha  b(t)^3 +2c}~.
\end{equa}
This means that $b$ is decreasing until the square root in
\eref{eq:dotbminusalpha} vanishes. This defines $b_*$, and then \eref{eq:t2alpha} is replaced by
\begin{equ}
  t_-=\int_{b_*}^{b(0)} \frac{\d b }{\sqrt{\FOURTHIRDS \alpha b^3 +2c}}~.
\end{equ}
This leads to \eref{eq:convalpha}.

The assertions under (iii) are a simple variant of (i) and (ii). The
difference is that because $g''(0)=0$, we find now that $c=-\TWOTHIRDS
\alpha f''(0)^3$, and $c\ne0$ by the assumption $f''(0)\ne0$.
Note that in this case, the positivity of  $\dot b(t)$
  follows from the second order ODE \eref{eq:bddalpha}, given that
$b(0) = f''(0)\neq 0$ and $\dot b(0)=g''(0)=0$.
The only remaining case is (iv), $g''(0)=f''(0)=0$,  which implies that
$b(0)= \dot b(0)=0$, hence, directly
leads to $b(t)=b(0)=0$ by \eref{eq:bddalpha}.
\end{proof}

\section{Finite time divergence}\label{sec:appb}

\fh{We adapt here the proof of \citep{rammaha} to the $1+1$ dimensional
context}. The proof actually works in the same way in higher
dimensions, for which it was already spelled out in \citep{rammaha,rammaha2}.
Since we deal here only with the fact that the equation will diverge
in finite time, it suffices to consider instead of $\alpha >0$, $\beta
>0$, the simpler form  
  \begin{equ}\label{equ:normalized}
    u_{tt}-u_{xx}= u_x^2~.
  \end{equ}
Indeed, if $v(x,t)$ solves $v_{tt}=\alpha (v_x)^2 +\beta v_{xx}$, then 
$u(x,t)=\frac{\beta }{\alpha } v(\beta ^{-1/2} x,t)$ satisfies
\eref{equ:normalized}. Note that we work
  in $\real$, and not, as happens in some simulations, in periodic
  boundary conditions.
  Assume the initial conditions are
  \begin{equa}
    u(x,0)&=f(x)~,\quad
    u_t(x,0)=g(x)~,
  \end{equa}
  with $f$, $g$ having support in $|x|<X$.
  We may assume $(f,g)\in H$ for some functional space, for example
  $H=W^{1,\infty}(\real)\times L^\infty(\real)$.\\
  Using a standard fixed point technique, we can find a local in time
  solution $u(t) \in C([0,t_0], H)$ for some small $t_0>0$. From the
  finite speed of propagation, using a cut-off technique, this global
  (in space) existence result extends to some local in time existence result in slices of
  backward cones of slope 1. This way, we can see that our solution is
  defined beyond the above-mentioned strip $\real \times [0,t_0]$, and extends to a larger domain of definition, which happens to be a union of backward light cones, with different heights. From elementary considerations, one of the following cases occurs:

  \medskip
  
(i) Either the union is the half-space $\real\times [0,\infty)$. We say in that case that the solution is ``global'' (for forward time).

\medskip

(ii) Or, the union writes as
\[
\{ (x,t) \; | \; 0\le t < T(x)\}
\]
for some 1-Lipschitz function $T : \real \to \real$. In that case, we say that
$u$ ``blows up in finite time.'' Note that by construction, we have a
local blow-up time for each $x\in \real$, namely $T(x)$. For more details
on the construction of the domain of definition, see \citep{alinhac1995} and
also \citep{azaiez_masmoudi_zaag_2019}.

  Define now, see \citep{rammaha}, for $x>0$,
  \begin{equ}
    M(x)=\HALF f(x)+\HALF \int_x^X g(\xi)\d\xi~.
  \end{equ}
 
\begin{theorem}\label{thm:t1}
Assume  there is an $X_0\in (0,X)$ for which $M(x)\ge0$ for all
  $x\in(X_0,X)$,
  and also
  \begin{equ}
    \int_{X_0}^X M(\xi)\d \xi \equiv\epsilon >0~.
  \end{equ}
  Then, the solution blows up in finite time, in the sense of the
  definitions above.
\end{theorem}
\begin{remark}
  Note that the theorem is shown under the assumption that
  $M(x)>0$. This covers cases where $u(x,0)$ and/or $u_t(x,0)$ are positive
  (or positive near the edge of their support).  The case of negative
  $M$ is not covered by the literature, nor by our proof. However, we
  have studied many cases with negative initial $M$. For example, the
  case $u(x,0)=0$ and $u_t(x,0)=-exp(-C\cdot x^2)$ of \fref{fig:smallbetadist}. In all these cases we seem
  to see ``M-type'' divergence. We can consider the solution $u(x,t)$,
  $u_t(x,t)$ as a new initial condition for any $t$. Then, we have
  observed that $M$ starts out negative, decreases, and finally crosses
  0. From that point on, we are again in the domain of validity of
  \tref{thm:t1}, and, indeed, the solution diverges. (This is reminiscent of the two cases in \aref{a:hatem}.)
\end{remark}
\begin{proof} It suffices to consider the situation where $u(x,0)$ and
  $u_t(x,0)$ have
support in $|x|\le X$, but  we will consider only the side of positive
$x$ in the sequel. Clearly, $u(x,t)=0$ for $x\ge t+X$ due to the
finite propagation speed.
One estimates now the function
\begin{equ}
  H(t)=\int_{X_1}^t (t-\tau )\int _{\tau +X_0}^{\tau +X} u(\xi,\tau
  )\d\xi \,\d\tau ~.
\end{equ}
Here, $X_1=(X-X_0)/2$.
From the definition, we get
\begin{equ}\label{eq:hpp}
  H''(t)=\int _{t +X_0}^{t +X} u(\xi,t)\d\xi~.
\end{equ}
One has the explicit formula
\begin{equ}\label{eq:uxt}
  u(x,t)=u_0(x,t)+\HALF\int_0^t \int_{x-t+\tau }^{x+t-\tau }
  u_x(\xi,\tau )^2 \d\tau \,\d\xi~,
\end{equ}
with the ``free evolution''
\begin{equ}
  u_0(x,t)=\HALF\left(f(x-t)+f(x+t)\right)+\HALF \int_{x-t}^{x+t} g(\xi)\d\xi~.
  \end{equ}
When $x\ge t+X_0$ and $X\ge t \ge X_1$,
then $x+t\ge X$ and therefore $f(x+t)=0$ and therefore, in this region,
\begin{equ}
  u_0(x,t)=\HALF f(x-t) + \HALF\int_{x
    -t}^{x+t} g(\xi) \d \xi~.
\end{equ}
We get from \eref{eq:uxt} and \eref{eq:hpp},
\begin{equ}
  H''(t)=G_0(t)+G_1(t)~,
\end{equ}
with
\begin{equ}
  G_0(t)= \int_{t+X_0}^{t+X}  u_0(x,t)\d x = \int_{t+X_0}^{t+X}
  M(x-t)\d x=\int_{X_0}^X M(x)\d x= \epsilon~,
\end{equ}
by the definition of $\epsilon $ in \tref{thm:t1}.
The nonlinearity leads to
\begin{equ}
  G_1(t)= \int_{t+X_0}^{t+X}\d x \int_{0}^{t}\d\tau
  \int_{x-t+\tau}^{x+t-\tau}\kern-2em\d\xi \,u_x(\xi,\tau)^2~.
\end{equ}
In \lref{lem:3} below, we show that for $t>X_1\equiv(X-X_0)/2$ one has
\begin{equa}\label{eq:lem2}
  G_1(t)\ge\frac{1}{t+X}\int_{0}^t\d\tau  \int_{\tau +X_0}^{\tau +X} \d\xi
  \,(t-\tau )\,(\xi-\tau -X_0) \,u_x(\xi,\tau)^2~.
\end{equa}

We use now the Schwarz inequality in the form
\begin{equ}
  \int \phi \psi = \int \phi^{1/2} (\phi^{1/2}\psi)\le \left(\int
  \phi\psi^2\right)^{1/2} \, \left(\int \phi\right)^{1/2}~,
\end{equ}
with $\phi=(t-\tau
)(\xi-\tau -X_0) $ and $\psi=u_x$. This leads to
\begin{equ}
  G_1(t)\ge F^2(t)/J(t)~,
  \end{equ}
with
\begin{equ}
F(t)=\int_0^t \int_{\tau +X_0}^{\tau +X} (t-\tau
  )(\xi-\tau -X_0) u_x(\xi,\tau )\d \xi \,\d\tau  
\end{equ}
and
\begin{equ}
  J(t)=\int_0^t \int_{\tau +X_0}^{\tau +X}  (t-\tau
  )(\xi-\tau -X_0)\d\tau \,\d\xi=\frac{(X-X_0)^2 t^2}{4}~.
\end{equ}
If we integrate the expression for $F$ by parts (in $\xi$),
we get
\begin{equ}
  F(t)=-\int_0^t \int_{\tau +X_0}^{\tau +X} (t-\tau
  ) u(\xi,\tau )\d \xi \,\d\tau  = -H(t)~.
\end{equ}
Therefore, we find finally
\begin{equ}\label{eq:final}
  H''(t)\ge G_0(t)+ \frac{H(t)^2}{J(t)}~.
\end{equ}
Fix now $T$ and we will show that the solution cannot exist
for
$T>T_*$, where $T_*$ will be computed in the proof:
We use here Lemma 1 from \citep{rammaha2} adapted to the 1d case.
The ingredients are that
\begin{equ}\label{eq:eps}
  H''(t)\ge G_0(t)=\epsilon >0~,
\end{equ}
for all $t\ge0$ and
\begin{equ}\label{eq:7}
  H''(t)\ge G_1(t)\ge 4\frac{H(t)^2}{(X-X_0)^2\,  t^2}~,
\end{equ}
for $t>X_1$ (as long as the solution exists). Furthermore, $H(X_1)=H'(X_1)=0$.

Fix now $T_1=2(X_1+1)$. Then, for $t>T_1$, we have $t>\HALF(t+1)$,
and we replace from now on \eref{eq:7} by the simpler
\begin{equ}\label{eq:Hpp}
  H''(t) \ge K_1 \frac{H(t)^2}{(t+1)^2}~, \text{ for } t>T_1~,
\end{equ}
with $K_1=16/(X-X_0)^2$.

The idea is now to deduce from \eref{eq:eps} and \eref{eq:Hpp} an
inequality of the form
\begin{equ}\label{eq:diverge}
H'(t)\ge C H^{1+\delta}(t)\text{  for  }t>T_1\text{ with
}\delta >0~.
\end{equ}
This implies divergence in finite time,
when $H(T_0)>0$. Indeed, if $H(T_0)=c^{-1/\delta }>0$, then
\begin{equ}\label{eq:diverge2}
  H(t)=\frac{1}{(c-C\delta (t-T_0))^{1/\delta }}~.
\end{equ}
One can reformulate this as follows:
If $H(T_0)=A$ and $A\le 1/e$, the optimising $\delta $ in
\eref{eq:diverge2} is $\le1$ and therefore
we find that the divergence time is proportional to $-\log(A)$.
Note that, if, for example, the leading edge of the support (at $x=0$) is like
$|x|^2$  for $x<0$, then this will lead to earlier divergence compared
to $|x|^3$.

We now begin the proof proper.
If $B>0$, we will use repeatedly the inequality
\begin{equ}\label{eq:half}
  \frac{x}{x+B}\ge \HALF ~, \text{ for all } x\ge B~.
\end{equ}
From \eref{eq:eps} we find
\begin{equ}\label{eq:181}
  H(t)\ge K_2 \epsilon t^2~, \text{ for all } t>0~,
\end{equ}
with $K_2=\HALF$.

Substituting \eref{eq:181} into \eref{eq:Hpp}, we get
\begin{equ}\label{eq:182}
  H''(t)\ge K_1K_2\epsilon  H(t) \frac{t^2}{(t+1)^2}\ge K_3 \epsilon
  H(t)~, \text { when } t>T_2~,
\end{equ}
for some large enough $T_2=\const T_1$ , not depending on $\epsilon $.
Since $H'(t)>0$, we can multiply \eref{eq:182} by $H'$ and write it as
\begin{equ}
 \frac{\d}{\d t}(H'(t)^2)\ge K_3 \epsilon  \frac{\d}{\d t}(H(t)^2) \text { when } t>T_2~.
\end{equ}
We integrate from $T_2$ to $t$ and obtain
\begin{equ}
  H'(t)^2 \ge  K_3\epsilon \left( H(t)^2+H'(T_2)^2-H(T_2)^2\right)=
K_3\epsilon H(t)^2 +K_4\epsilon 
  \text { when } t>T_2~,
\end{equ}
for some $K_4$.
From \eref{eq:eps}, we conclude that for large enough $T_3$, one has
\begin{equ}
K_3\epsilon H(t)^2 +K_4\epsilon \ge K_5\epsilon  H(t)^2~,\text{ when } t>T_3~.
\end{equ}
Combining the last two equations we find
\begin{equ}
  H'(t)\ge K_5^{1/2} \epsilon ^{1/2}  H(t),~
  \text { when } t>T_3~.
\end{equ}
Integrating from $T_3$ to $t$ leads to
\begin{equ}\label{eq:exp}
  H(t)\ge H(T_3)\exp\left(K_6 \epsilon ^{1/2} (t-T_3)\right)\ge
  H(T_3)\exp\left(\HALF K_6 \epsilon ^{1/2}t\right)~,\text{ when }t>2T_3, 
\end{equ}
with $K_6=K_5^{1/2}$.
Substituting again into \eref{eq:Hpp}, we get
\begin{equ}
  H''(t) \ge K_7 H(t)^{1+\delta } ~,\text{ for any }\delta >0~,
\end{equ}
since the exponential in \eref{eq:exp} (to the power $\delta>0 $) will
dominate the factor $(t+1)^{-2}$ of \eref{eq:Hpp}, only if $t$ is
sufficiently large. (Note that $K_7$ and this new minimal time $T_4$ will
depend on $\delta $.)
We now multiply the last equation by $H'$ and we obtain
\begin{equ}
  \frac{\d}{\d t}(H'(t)^2) \ge  \frac{2K_7}{2+\delta  }\frac{\d}{\d
    t}\left( H(t)^{2+\delta } \right)~,\text{ for } t>T_4~.
\end{equ}
Integrating from $T_4$ to $t$
we find
\begin{equ}
  H'(t)^2 \ge  \frac{2K_1}{2+\delta  } \left(H(t)^{2+\delta
  }-H(T_{4})^{2+\delta }\right)+H'(T_{4})^2~.
\end{equ}
Taking square roots on both sides and choosing $T_*$
sufficiently larger than $T_4$, we finally arrive at \eref{eq:diverge}
from which we see that there is a divergence in finite time, as in \eref{eq:diverge2}.
\end{proof}

We still need to show the inequality \eref{eq:lem2}.
\begin{lemma}\label{lem:3}
  Let
 \begin{equa}
  G_1(t)&= \int_{t+X_0}^{t+X}\d x \int_{0}^{t}\d\tau
  \int_{x-t+\tau}^{x+t-\tau}\kern-2em\d\xi \,u_x(\xi,\tau)^2~.
 \end{equa}
Let $X>X_0>0$, and assume $u_x(x,t)=0$ for all $|x|\ge t+X$, $0\le
t\le T$. Then one has for all $t\ge X_1\equiv(X-X_0)/2$ the inequality
 \begin{equa}\label{eq:ineq}
  G_1(t)\ge\frac{1}{t+X}\int_{0}^t\d\tau  \int_{\tau +X_0}^{\tau +X} \d\xi
  \,(t-\tau )\,(\xi-\tau -X_0) \,u_x(\xi,\tau)^2~.
\end{equa}
\end{lemma}
\begin{proof}
  We first decompose the triple integration into 3 pieces:
  The original integration is over the domain
  \begin{equs}[r0]
    X_0+t\,&\le x \,&\le\,& X+t~,\\
    0\,&\le \tau\,&\le \,&t~,\\
    x-t+\tau \,&\le \xi \,&\le \,&x+t-\tau ~.
  \end{equs}
  Also note that the integrand has support in $\xi\le\tau+X$.
  The three pieces are
  \begin{equs}[r1]
    0\,&\le \tau\,&\le \,&t-X_0~,\\
   \tau +X_0\,&\le \xi\,&\le \,&\tau +X~,\\
        t+X_0\,&\le x\,&\le \,&\xi+t-\tau ~,
  \end{equs}
  and
  \begin{equs}[r2]
    t-X_0\,&\le \tau\,&\le \,&t~,\\
   \tau +X_0\,&\le \xi\,&\le \,&2t-\tau +X_0~,\\
        t+X_0\,&\le x\,&\le \,&\xi+t-\tau ~,
  \end{equs}
  and
  \begin{equs}[r3]
    t-X_0\,&\le \tau\,&\le \,&t~,\\
   2t-\tau +X_0\,&\le \xi\,&\le \,&\tau +x~,\\
       \xi-t+\tau\,&\le x\,&\le \,&\xi+t-\tau ~.
  \end{equs}
One can show that
\eref{r1}--\eref{r3} defines a domain which coincides with that of \eref{r0}, and that the 3
regions are disjoint.

We now give lower bounds for the 3 regions.
For \eref{r1} we get
\begin{equa}
  &\int_{0}^{t-X_1}\d\tau
  \int_{\tau +X_0}^{\tau +X}\d\xi
  \,u_x(\xi,\tau)^2~\int_{t+X_0}^{\xi+t-\tau}  \d x ~,\\
  &\int_{0}^{t-X_1}\d\tau
  \int_{\tau +X_0}^{\tau +X}\d\xi
  \,u_x(\xi,\tau)^2~ (\xi-\tau -X_0) ~.
\end{equa}
We bound the last factor from below by
\begin{equa}
  (\xi-\tau -X_0)\ge \frac{t-\tau }{t+X}(  \xi-\tau -X_0)~.
\end{equa}
Similarly, for \eref{r2}, the $x$ integration is bounded from below in
exactly the same way.
Finally, for \eref{r3}, using the support property $\xi\le\tau +X$, we
find (since $t>X_1$ and $X_1<X$),
\begin{equa}
 \xi-\tau -X_0 \le X-X_0=2X_1<t+X_1< t+X~.
\end{equa}
This leads to a bound for the $x$ integral of the form
\begin{equa}
  \int_{\xi-t+\tau }^{\xi+t-\tau } \d x=2(t-\tau )\ge (t-\tau
  )\frac{\xi-\tau -X_0}{t+X}~.
\end{equa}

Collecting terms, we finally find that
\begin{equa}
  G_1(t)\ge\frac{1}{t+X}\int_{0}^t\d\tau  \int_{\tau +X_0}^{\tau +X} \d\xi
  \,(t-\tau )\,(\xi-\tau -X_0) \,u_x(\xi,\tau)^2~.
\end{equa}
\end{proof}

\section{Numerics}\label{a:numerics}
We integrate all the equations by using the Dorman-Prince
\cite{Ehairer} Runge-Kutta integrator. The functions are discretised
in $2^{13}$ equidistant points. Derivatives are computed by
using 5-point stencils. Divergence is defined by
$\max(|u_{xx}|)>10^6$.
Special care has been given to assert the quality of the results: We
compute with a tolerance (if achievable) of $10^{-11}$. There are two
situations where the integration can fail: The time steps gets too
short (this happens sometimes when $\alpha $ is large and $\beta $ is
small).
The other problem is the size of the domain in $x$: We take periodic
boundary conditions on $[-\pi,\pi]$, and initial data which vanish at
these boundaries. If, during time evolution, the value of $|u(\pm\pi,t)|$
exceeds $10^{-4}$, we consider that the wave-part of the evolution has
``hit'' the boundary, and we stop the calculation. This happens
especially if $\beta $ is large and $\alpha $ is small, because in
this case, the wave moves with speed $\sqrt{\beta }$, and may hit the
boundary \emph{before} $\alpha $ can lead to a divergence.

\begin{conjecture}
  We believe that this phenomenon might account for the idea that
  large $\beta$ regularises the PDE in cosmological simulations as
  discussed in \cite{hassani_JCAP}. But, as we 
  show in \sref{sec:appb} the mathematical fact is that all solutions
  diverge in a finite time (unless they are 0). However, this will require a detailed study in a cosmological context.
\end{conjecture}
\end{appendices}
\section*{Acknowledgements} \fh{We thank Edriss Titi to have brought us together for this paper.}
We profited from discussions with Julian Adamek, Pierre Collet, Martin
Kunz, Sabir Ramazanov, Pan Shi, Peter Wittwer, and  Alexander Vikman. 
This work was partially supported by an ERC advanced grant (Bridges 290843) and by SNF Swissmap.
\section*{References}
\bibliographystyle{nonlinearity}
\bibliography{bibliography}
\end{document}